\documentclass[11pt]{amsart}

\usepackage{amssymb}
\usepackage[mathscr]{eucal}
\usepackage{graphicx}
\usepackage{mathrsfs}
\usepackage{psfrag}
\usepackage{fullpage}
\usepackage{color}

\theoremstyle{plain}
\newtheorem{theorem}{Theorem}[section]

\newtheorem{assumption}[theorem]{Assumption}

\newcommand{\lb}{\left\{}
\newcommand{\rb}{\right\}}

\newcommand{\R}{\mathbb{R}}
\newcommand{\N}{\mathbb{N}}

\newcommand{\eps}{\varepsilon}
\newcommand{\haz}{\lambda}

\newcommand{\Borel}{\mathscr{B}}
\newcommand{\Pspace}{\mathscr{P}}
\newcommand{\BP}{\mathbb{P}}
\newcommand{\BQ}{\mathbb{Q}}
\newcommand{\BE}{\mathbb{E}}
\newcommand{\filt}{\mathscr{F}}

\newcommand{\la}{\left \langle}
\newcommand{\ra}{\right\rangle}

\newcommand{\ee}{\mathfrak{e}}

\newcommand{\genL}{\mathcal{L}}

\newcommand{\pp}{\mathsf{p}}
\newcommand{\pt}{\star}

\newcommand{\Fail}{L}

\newcommand{\Types}{\mathcal{P}}
\newcommand{\PP}{\mathcal{P}}

\newcommand{\TInt}{\mathcal{T}}
\newcommand{\NN}{n}
\newcommand{\mart}{\mathcal{M}}

\newcommand{\QQ}{\mathcal{Q}}

\allowdisplaybreaks

\begin{document}
\title{Systemic Risk and Default Clustering for Large Financial Systems}

\author{Konstantinos Spiliopoulos}
\address{Department of Mathematics \& Statistics\\
Boston University\\
Boston, MA 02215}
\email{kspiliop@math.bu.edu}

\date{\today.}

\begin{abstract}
As it is known in the finance risk and macroeconomics literature, risk-sharing in large portfolios may increase the probability of creation of
 default clusters and of systemic risk. We review recent developments on mathematical and computational tools for the quantification of such phenomena.  Limiting analysis such as law of large numbers and central limit theorems allow to approximate the distribution in large systems and study quantities such as the loss distribution in large portfolios. Large deviations analysis allow us to study the tail of the loss distribution and to identify pathways to default clustering. Sensitivity analysis allows to understand  the most likely ways in which different effects, such as contagion and systematic risks, combine to lead to large default rates.  Such results could give useful insights into how to optimally safeguard against such events.

\end{abstract}

\maketitle

\section{Introduction}

The past several years have made clear the need to better understand the behaviour in large interconnected financial systems. Almost all areas of modern life are touched by a financial crisis. The recent financial crisis of $2007-2009$ brought into focus the networked structure of the financial world. It challenged the mathematical finance community to understand connectedness in financial systems. The understanding of systemic risk, i.e., the risk that a large numbers of components of an interconnected financial system fails within a short time leading to the failure of the system itself, becomes an important issue to investigate.

Interconnections often make a system robust, but they can also act as conduits for risk. Even things that may seemingly be unrelated, may become related as risk restrictions, may for example, force a sale of one type of a well-performing asset to compensate for the poor behavior of another asset. Thus, appropriate mathematical models need to be developed, in order to help in the understanding of how risk can propagate between financial objects.

It is possible that initial shocks could trigger contagion effects (e.g., \cite{Meinerding2012}). Examples of such shocks include:
changes in interest rate values, in currencies values, changes of commodities prices, or reduction in global economic growth. Then, there may be a transmission
 mechanism which causes other institutions in the system to be affected by the initial shock. An example of such a mechanism is financial linkages among economies.
Another reason could  simply be investor irrationality. In either case, systemic risk causes the perceived risk-return trade-off in the economy to change. Uncertainty becomes an issue and market participants fear subsequent losses in asset prices with a large dispersion in regards to the magnitude of the crisis. Reduce-form point process models of correlated default are many times used (a): to assess portfolio credit risk and (b): to value securities exposed to correlated default risk. The workhorses of these models are counting processes. In this work we focus on using dynamic portfolio credit risk models to study large portfolio asymptotics and default clustering.

Large portfolio asymptotic were first studied in \cite{vasicek}. The model in \cite{vasicek} is a  static model of a homogeneous
 pool and firms default independently of one another conditional on a normally distributed random variable representing a systematic risk factor.
Alternative distributions of the systematic factor were examined in \cite{schloegl-okane}, \cite{lucas-etal} and the case of heterogeneous portfolios was studied
in  \cite{Gordy03arisk-factor}. In \cite{hambly}, the authors extend the model of \cite{vasicek} dynamically and the systematic risk factor follows a Brownian motion. In \cite{hambly}, the authors study a structural model for distance to default process in a pool of names. A firm defaults when the default process hits zero.
 Exploiting conditional independence of defaults, \cite{ddd} and \cite{glasserman-kang} have studied the tail of the loss distribution in the static case. Large deviations arguments were also used in \cite{SpiliopoulosSowers2011} to study stochastic recovery effects on large static pools of credit assets.

Reduced-form models of correlated default timing have appeared in the finance literature under different forms.  \cite{giesecke-weber} take the intensity of a name as a
 function of the state of the names in a specified neighborhood of that name. The authors in \cite{daipra-etal} and \cite{daipra-tolotti} take the intensity to be
 a function of the portfolio loss and each name can be either in a good or in a distressed financial state. These papers prove law of large numbers for the portfolio loss distribution and develop Gaussian approximations to the portfolio
 loss distribution based on central limit theorems. \cite{CMZ} consider the typical behavior of a mean field system with permanent default impact.

\cite{sircar-zari}
 study large portfolio asymptotics for utility indifference valuation of securities exposed to the losses in the pool.
In \cite{PapanicolaouSystemicRisk}, the authors study systematic risk via a mean field model of interacting agents. Using a model of a two well potential, agents can move freely from a healthy state to a failed state.  The authors study probabilities of transition from the healthy to the failed state  using large deviations ideas. In \cite{FouqueIchiba2013} the authors propose and study a model for inter-bank lending and study its stochastic stability.

The authors in \cite{SahaliaDiazLaeven2010} employ
jump-diffusion models driven by Hawkes processes to empirically study default clustering and the time dimension of systemic risk. \cite{duan} proposes a hierarchical
model with individual shocks and group specific  shocks. The work of \cite{BieleckiCrepeyHerbertson} reviews intensity models that are governed by exogenous and endogenous
 Markov Chains. In \cite{GieseckeSpiliopoulosSowers2011}, the authors proposed a dynamic point process model of correlated default timing in a portfolio of firms (``names''). The model
incorporates different sources of default clustering identified in recent empirical research, including idiosyncratic risks, exposure to systematic risk factors
 and contagion in financial markets, see \cite{duffie-saita-wang}, \cite{azizpour-giesecke-schwenkler}. Based on the weak convergence ideas of \cite{GieseckeSpiliopoulosSowers2011}, the authors in \cite{BoCapponi2013} obtain and study formulas for the bilateral counterparty valuation adjustment of a credit default swaps portfolio referencing an asymptotically large number of entities.

The model in \cite{GieseckeSpiliopoulosSowers2011} can be naturally understood as an interacting particle system that is influenced by an exogenous source of randomness. There is a central source of interconnections and failure of any of the components stresses the central 'bus', which in turn can cause the failure of other components (a contagion effect). Computing the distribution of the loss from default
in such models tends to be a difficult task and while Monte-Carlo simulation methods are broadly applicable, they can be slow for large portfolios or large time horizons as it is commonly the interest in practice. Mathematical and computational tools for the approximation to the distribution of the loss from default in large heterogeneous portfolios were then developed in \cite{GieseckeSpiliopoulosSowersSirigano2012}, Gaussian correction theory was developed in \cite{SpiliopoulosSiriganoGiesecke2013} and analysis of tail events and most likely paths to failure via the lens of large deviations theory was then developed in \cite{SpiliopoulosSowers2013}. We remark here that to a large extend systemic risk refers to the tail of the distribution. The authors in \cite{SirignanoSchwenklerGiesecke2013} combine the large pool asymptotic results of \cite{GieseckeSpiliopoulosSowersSirigano2012}-\cite{SpiliopoulosSiriganoGiesecke2013} with maximum likelihood ideas to construct tractable statistical inference procedures for parameter estimation in large financial systems.

Such mathematical results lead to new computational tools for the measurement and prediction of risk in high-dimensional financial networks. These tools mainly include approximations of the distribution of losses from defaults and of portfolio risk measures, and  efficient computational tools for the analysis of extreme default events. The mathematical results  also yield {important insights} into the behavior of systemic risk as a function of the characteristics of the names in the system, and in particular their interaction.

Financial institutions (banks, pension funds, etc) often hold large portfolios in order to diversify away a number of idiosyncratic effects of individual assets.  Deposit insurance premia  depend upon meaningful models and assessment of the macroeconomic effect of the various phenomena that drive defaults.
Development of related mathematical and computational tools can help  inform the design of regulatory policy, improve the pricing of federal deposit insurance, and lead to more accurate risk measurement at financial institutions.

In this paper, we focus on dynamic default timing models for large financial systems that fall into the category of intensity models in portfolio credit risk. Based on the default timing model developed in \cite{GieseckeSpiliopoulosSowers2011}, we address several of the issues just mentioned and that are typically of interest. The mathematical and computational tools developed allow to reach to financial related conclusions for the behavior of such large financial systems.

Although the primary interest of this work is risk in financial systems, models of the type discussed in this paper are generic enough to allow for modifications that make them relevant in other domains, including systems reliability, insurance and epidemiology.    In reliability, a large system of interacting components might have a central connection, and be influenced by an external environment (temperature, for example). The failure of an individual component (which could be governed by an intensity model appropriate for the particular application) increases the stress on the central connection and thus the other components, making the entire system more likely to fail.
In insurance,  the system could represent a pool of insurance policies.  The effect
of wildfires might, in that example, be modelled by a contagion term.  Systematic risk in the form of environmental conditions has an impact on the whole pool.

The rest of the article is structured as follows. In Section \ref{S:Model} we describe the correlated default timing proposed in \cite{GieseckeSpiliopoulosSowers2011}. Section \ref{S:LLN} studies the typical behavior of the loss distribution in such portfolios as the number of names (agents) in the pool grow to infinity. Section \ref{S:CLT} focuses on developing the Gaussian correction theory. As we shall see there, Gaussian corrections are very useful because they make the approximations accurate even for portfolios of relatively small sizes. In Section \ref{S:LDP}, we study the tail of the loss distribution using arguments from the large deviations theory. We also study the most likely path to systemic failure and to the development of default clusters. An understanding of the preferred paths to large default rates and the most likely path to the creation of default clusters can give useful insights into how to optimally safeguard against such events. Importance sampling techniques can then be used to construct asymptotically efficient estimators for tail event probabilities, see Section \ref{S:IS}.  Conclusions are in Section \ref{S:Conclusions}. A large part of the material presented in this work, but not all, is related to recent work of the author described in \cite{GieseckeSpiliopoulosSowers2011}, \cite{GieseckeSpiliopoulosSowersSirigano2012}, \cite{SpiliopoulosSiriganoGiesecke2013} and \cite{SpiliopoulosSowers2013}.

\section{A dynamic correlated default timing model}\label{S:Model}

One of the issues of fundamental importance in financial markets is systemic risk, which may be understood as the likelihood of failure of a substantial fraction of firms in the economy.
There are a number of ways of interpreting this, but our focus will be the behavior of actual \emph{defaults}.  Defaults are discrete events, so one can frame the interest within
the language of point processes.   Empirically, defaults tend to happen in groups; feedback and exposure to market forces (along the lines of ``regimes'') tend to produce
correlation among defaults.

Let us fix a probability space $(\Omega,\filt,\BP)$ where all random variables will be defined. Denote by $\tau^\NN$ the stopping time at which the $n$-th component (or particle) in our system fails. Then, as $\delta\searrow 0$, a failure time $\tau^\NN$ has intensity process $\lambda^{\NN}$, which satisfies
\begin{equation}\label{E:stochhaz} \BP\{\tau^\NN\in (t,t+\delta]|\filt_t,\, \tau^\NN>t\} \approx \lambda^\NN_t \delta, \end{equation}
where $\filt_t$ is the sigma-algebra generated by the entire system
up to time $t$. Hence, we essentially have that the process defined by $1_{\{\tau^\NN\le t\}}-\int^{t}_0\lambda^\NN_s 1_{\{\tau^\NN>s\}}ds$ is a martingale.

Motivated by the empirical studies in \cite{duffie-saita-wang} and \cite{azizpour-giesecke-schwenkler}, we may model the intensity $\lambda^{\NN}$ in such a way
that it depends on three factors: a mean reverting idiosyncratic source of risk, the portfolio loss rate and a systematic risk factor. Heterogeneity can  be addressed
by allowing the intensity parameters of each name to be different. The mean reverting character of the idiosyncratic source of risk is there to guarantee that the effect
 of a default in the pool has a transient effect on the default intensities of the surviving names. The dependence on the portfolio loss rate, denoted by $L^{N}_{\cdot}$
is the term that is responsible for the contagious effects, whereas the systematic risk factor, denoted by $X_{\cdot}$ is an exogenous source of risk.  To be precise, the
 default intensities, $\lambda^{\NN}$'s, are governed by the following interacting system of stochastic differential equations (SDEs)

\begin{equation} \label{E:maina}
\begin{aligned}
d\lambda^\NN_t &= -\alpha_\NN (\lambda^\NN_t-\bar \lambda_\NN)dt + \sigma_\NN \sqrt{\lambda^\NN_t}dW^n_t +  \beta^C_\NN dL^N_t+ \eps \beta^S_\NN \lambda^\NN_t dX_t,\quad
\lambda^\NN_0 = \lambda_\circ^\NN.
\end{aligned}
\end{equation}
where, $\{W^n\}_{n\in \N}$ be a countable collection of independent standard Brownian motions.

The process $L^N_t$ represents the empirical failure rate in the system, i.e.,
\begin{equation} \label{E:mainc}
 L^N_t =\frac{1}{N}\sum_{n=1}^N 1_{\{\tau^\NN\le t\}},
 \end{equation}
 where by letting
$\{\ee_n\}_{n\in \N}$ to be an i.i.d. collection of standard exponential random variables we have
\begin{equation} \label{E:maind} \tau^\NN = \inf\left\{ t\ge 0: \int_{s=0}^t \lambda^\NN_s ds\ge \ee_n\right\}.\end{equation}

The process $X_{t}$ represents the systematic risk, which can be modeled to be the solution to some SDE
\begin{equation}\label{E:mainb}
\begin{aligned} dX_t &= b_{0}( X_t) dt + \sigma_{0}(X_t)dV_t,\quad
X_0= x_\circ.
\end{aligned}
\end{equation}
where  $V$ is a standard Brownian motion which is independent of the $W^n$'s and $\ee_n$'s. Plausible models for $X_{t}$ could be an Ornstein-Uhlenbeck process or a Cox-Ingersoll-Ross (CIR)
 process.

In the case $\beta^{C}_{\NN}=\beta^{S}_{\NN}=0$ for all $n\in\{1,\cdots,N\}$, one recovers the classical CIR process model in credit risk, e.g., \cite{dps}. Namely, the intensity SDE \eqref{E:stochhaz} extends the widely-used CIR process by including two additional terms that generate correlation between failure times. The term  $\eps \beta^S_\NN \lambda^\NN_t dX_t$ induces correlated diffusive movements of the component intensities; the process $X$ represents the state of the macro-economy, which affects all assets in the pool. The term $\beta^C_\NN dL^N_t$ introduces a feedback (contagion) effect. The standard term $-\alpha_\NN (\lambda^\NN_t-\bar \lambda_\NN)dt$ is a mean reverting term allowing
the component to ``heal'' after a shock (i.e., a failure).   This parsimonious formulation allows us to take advantage of the wealth of knowledge about CIR-type processes. The parameter $\eps>0$ allows us to later on focus on rare events.

The process $L^N$ of \eqref{E:mainc}, which simply gives us the fraction of components which have already failed by time $t$, affects
each of the remaining components in a natural way.  Each failure corresponds to a Dirac function in the measure $dL^N$; the term $\beta^C_\NN dL^N_t$ thus
leads to upward impulses in $\lambda^\NN$'s, which leads (via \eqref{E:maind}) to sooner failure of the remaining functioning components.
We might think of a central ``bus'' in a system of components.  Each of the components
depends on this bus, which in turn sensitive to failures in the various components.
In the financial application that was considered in \cite{GieseckeSpiliopoulosSowers2011},
this feedback mechanism is empirically observed to be an important channel for the clustering of defaults in the U.S. (see \cite{azizpour-giesecke-schwenkler}).

In order to allow for heterogeneity, the parameters in \eqref{E:maina} depend on the index $n$.  Define the ``type''
\begin{equation}\label{E:ppdef} \pp^\NN_{t} = (\lambda^{\NN}_{t},\alpha_\NN,\bar \lambda_\NN, \sigma_\NN,\beta^C_\NN,\beta^S_\NN) \end{equation}
for each $n\in \N$ and $t\geq 0$. The $\pp^\NN_{t}$'s take value in $\PP=\R^{3}_{+}\times\R\times\R_{+}\times\R\subset\R^{6}$. The parameters $(\lambda^{\NN}_{0},\alpha_\NN,\bar \lambda_\NN, \sigma_\NN,\beta^C_\NN,\beta^S_\NN)$ are assumed to be bounded uniformly in $n\in\N$.

 We can capture the heterogeneity of the system by defining
$U_{N}=\tfrac1{N}\sum_{n=1}^N \delta_{\pp_\NN}$ and assuming that this empirical type frequency has a (weak) limit. In particular we make the following assumption
\begin{assumption}\label{A:regularity}
We assume that $U= \lim_{N\to \infty}U_N$ exists \textup{(}in $\Pspace( \Types)$\textup{)}.
 \end{assumption}

Proposition 3.3 in \cite{GieseckeSpiliopoulosSowers2011} guarantees that under the assumption of an existence of a unique strong solution for the SDE for $X_{\cdot}$
process, the system  \eqref{E:maina}--\eqref{E:mainb} has a unique strong solution such that $\lambda^{\NN}_{t}\geq 0$ for every $N\in\N$, $n\in\{1,\cdots,N\}$ and $t\geq 0$. The model \eqref{E:maina}--\eqref{E:mainb} is a mean-field type model; the feedback occurs through the empirical average of the pool of names.  It is somewhat similar to certain genetic models (most notably the Fleming-Viot process; see \cite{DawsonHochberg}, \cite[Chapter 10]{MR88a:60130},  and \cite{FlemingViot}). However, as it is also demonstrated in \cite{GieseckeSpiliopoulosSowers2011} and in \cite{GieseckeSpiliopoulosSowersSirigano2012}, the structure of the system  \eqref{E:maina}--\eqref{E:mainb} presents several difficulties
that bring the analysis of such systems outside the scope of the standard setup.

\section{Typical behavior: Law of large numbers}\label{S:LLN}

The system \eqref{E:maina}--\eqref{E:mainb} can
naturally be understood as an interacting particle system.  This
suggests how to understand its large-scale behavior.  The structure of the feedback
(the empirical average $\Fail^N$) is of mean-field type (roughly within the class of McKean-Vlasov models; see \cite{Gartner88}, \cite{KotolenezKurtz2010}).
 An understanding of ``typical'' behavior of a system as $N\to \infty$ is fundamental in identifying ``atypical'' or ``rare'' events.

To formulate the law of large numbers result, we define the empirical  distribution of the $\pp^\NN$'s corresponding to the names that have survived up to time $t$, as follows:
\begin{equation*} \mu^N_t = \frac{1}{N}\sum_{n=1}^N\delta_{\pp^N_t}1_{\{\tau^\NN>t\}}.
\end{equation*}
This captures the entire dynamics of the model (including the effect of the heterogeneities).  We can directly calculate the failure rate from the $\mu^N$'s:
\begin{equation}\label{Eq:PortfolioLoss}
 \Fail^N_t =1-\mu^N_t(\PP),\qquad t\ge 0.
\end{equation}
Let us then identify the limit of $\mu^N_t(\PP)$ as $N\to \infty$. This is a law of large numbers (LLN) result and it
identifies the baseline ``typical'' behavior of the system.
For $f\in C^2(\PP)$, let
\begin{equation}\label{E:Operators1}
\begin{gathered} (\genL_1 f)(\pp) = \frac12 \sigma^2\haz\frac{\partial^2 f}{\partial \haz^2}(\pp) - \alpha(\haz-\bar \haz)\frac{\partial f}{\partial \haz}(\pp)-\haz f(\pp)\\
(\genL_2 f)(\pp) = \beta^C\frac{\partial f}{\partial \haz}(\pp)\\
(\genL_3^x f)(\pp) = \eps\beta^S \haz b_{0}(x)\frac{\partial f}{\partial \haz}(\pp)+\frac{\eps^2}{2}(\beta^S)^2\haz^2\sigma_{0}^{2}(x)\frac{\partial^2 f}{\partial \haz^2}(\pp)\\
(\genL_4^x f)(\pp) =\eps\beta^S\haz\sigma_{0}(x)\frac{\partial f}{\partial
\haz}(\pp)\qquad \text{and}\qquad
\QQ(\pp) = \haz
 \end{gathered}
\end{equation}
for $\pp = (\haz,\alpha,\bar \haz,\sigma,\beta^C,\beta^S)$.
The generator $\genL_1$ corresponds to the diffusive part of the
intensity with killing rate $\haz$, and $\genL_2$ is the
macroscopic effect of contagion on the surviving intensities at any
given time.  The operators $\genL_3^x$ and $\genL_4^x$ capture the dynamics due to the exogenous systematic risk $X$.  Then $\mu^N$ tends in distribution (in the natural topology
of subprobability measures on $\PP$) to a measure-valued process $\bar \mu$.  Letting
\begin{equation*} \la f,\mu\ra = \int_{\pp\in \PP}f(\pp)\mu(d\pp) \end{equation*}
for all $f\in C^2(\PP)$, the limit $\bar \mu$ satisfies the stochastic evolution equation
\begin{equation}\label{E:weakSIPDE}
d\la f,\bar \mu_t\ra =  \left\{\la \genL_1f,\bar \mu_t\ra+ \la \QQ,\bar \mu_t\ra
\la \genL_2f,\bar \mu_t\ra+\la \genL^{X_t}_3 f,\bar
\mu_t\ra\right\}dt+\la \genL^{X_t}_4 f,\bar \mu_t\ra dV_t\quad
 \text{ a.s.}
\end{equation}
With sufficient regularity, this is equivalent to the stochastic integro-partial differential equation (SIPDE)
\begin{equation}\label{E:SIPDE}
d\upsilon =  \genL^*_1 \upsilon dt+ \left(\int \QQ \upsilon\right) \genL^*_2 \upsilon dt +\genL^{X_t,*}_3\upsilon dt+\eps \genL^{X_t,*}_4\upsilon dV_t\quad \text{ a.s.}
\end{equation}
where $^*$ denotes adjoint in the appropriate sense (for notational simplicity, we have written \eqref{E:SIPDE} to include the types as one of the coordinates; in a heterogeneous
collection in practice we would often use only $\haz$ in solving \eqref{E:SIPDE}). We recall the rigorous statement in Theorem \ref{T:MainLLN0}.

The SIPDE \eqref{E:SIPDE} gives us a ``large system approximation'' of the failure rate:
\begin{align}\label{E:firstorder}
\Fail^N_t\approx 1-\bar\mu_t(\PP)=1-\int_{\PP} \upsilon(t,\pp)d\pp.
\end{align}
The computation of the first-order approximation \eqref{E:firstorder} suggested by the LLN requires solving the SIPDE \eqref{E:SIPDE} governing the density of the limiting measure. In \cite{GieseckeSpiliopoulosSowersSirigano2012} a numerical method for this purpose is proposed. The method is based on an infinite system of SDE's for certain moments of
the limiting measure. These SDEs are driven by the systematic risk process $X$ and a truncated system can be solved using a discretization or random ODE scheme. The solution to the SDE system leads to the solution to the SIPDE via an inverse moment problem.

\begin{figure}[t!]
\begin{center}
\includegraphics[scale=0.5]{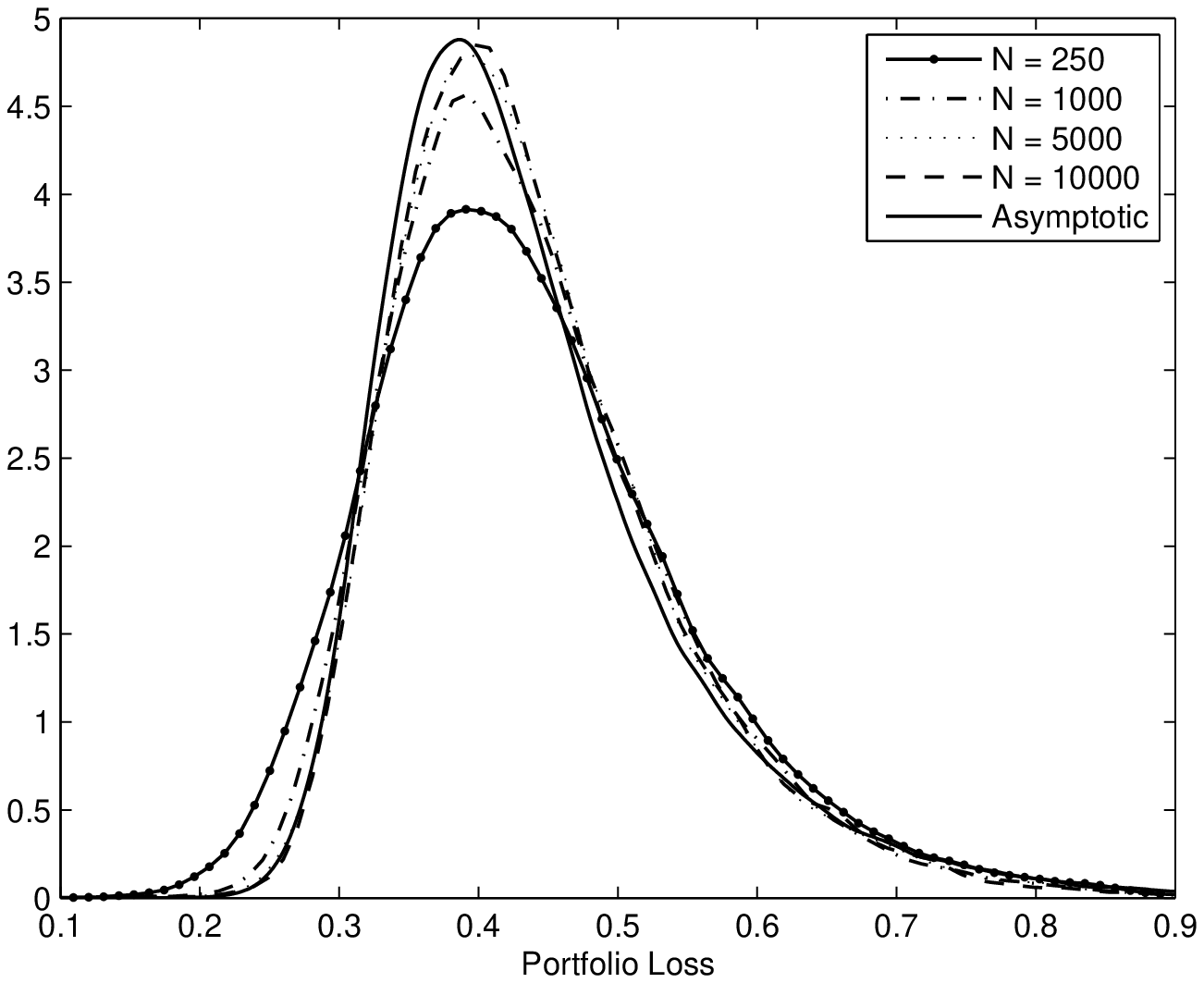}
\includegraphics[scale=0.5]{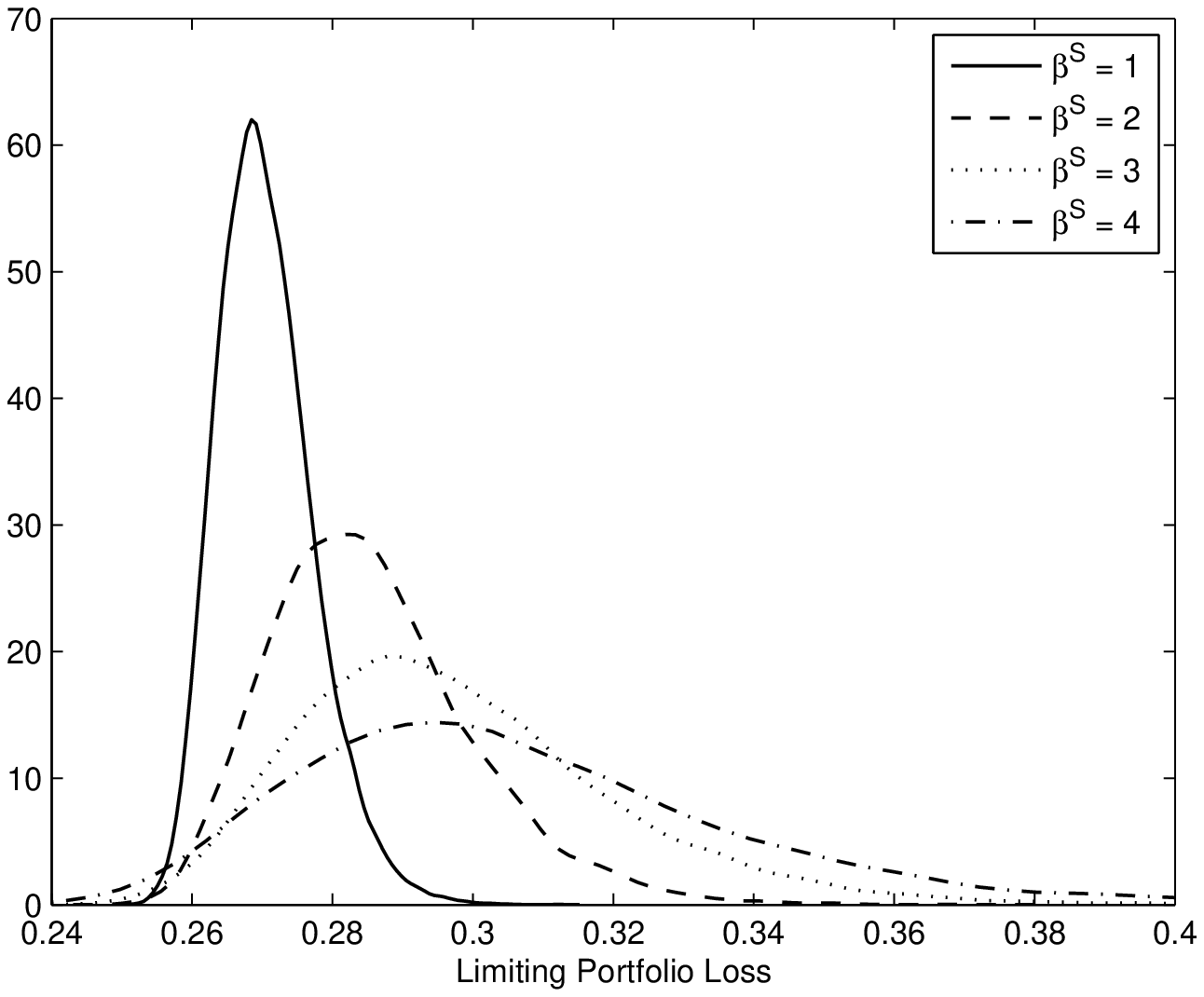}
\end{center}
\caption{\label{fig: BetaCone} On the left: Comparison of distributions of failure rate $\Fail^N_t$ for different $N$ at $t = 1$. Parameter choices:
$(\sigma,\alpha,\bar{\haz},\haz_0,\beta^C,\beta^S)=(.9,4,.2,.2,4,8)$. On the right: Comparison of distribution of limiting failure rate $1-\bar\mu_t(\PP) $ for different values of the systematic risk sensitivity $\beta^S$ at $t = 1.$ Parameter choices:
$(\sigma,\alpha,\bar{\haz},\haz_0,\beta^C)=(.9,4,.2,.2,2)$.
 }
\end{figure}

The approximation \eqref{E:firstorder} has significant computational advantages over a naive Monte Carlo simulation of the high-dimensional original stochastic system \eqref{E:maina}--\eqref{E:mainb} and its accuracy is demonstrated in the left of Figure  \ref{fig: BetaCone} for a specific choice of parameters. It also provides information about catastrophic
failure.

The tail represents extreme default scenarios, and these are at the center of risk measurement and management applications in practice. The analysis of the limiting distribution generates important insights into the behavior of the tails as a function of the characteristics of the system \eqref{E:maina}--\eqref{E:mainb}. For example, we see that the tail is heavily influenced by the sensitivity of a name to the variations of the systematic risk $X$. The bigger the sensitivity the fatter the tail, and the larger the likelihood of large losses in the system (see the right of Figure \ref{fig: BetaCone}).  Insights of this type can help understand the role of contagion and systematic risk, and how they interact to produce atypically large failure rates. This, in turn, leads to ways to minimize or ``manage'' catastrophic failures.

Let us next present the statement of the mathematical result. We denote by $E$  the collection of sub-probability measures (i.e., defective probability measures) on $\PP$; i.e., $E$ consists
of those Borel measures $\nu$ on $\PP$ such that $\nu(\PP)\le 1$.

\begin{theorem}[Theorem 3.1 in \cite{GieseckeSpiliopoulosSowersSirigano2012}]\label{T:MainLLN0}
We have that $\mu^{N}_{\cdot}$ converges in distribution to $\bar{\mu}_{\cdot}$ in
$D_E[0,T]$. The evolution of
$\bar{\mu}_{\cdot}$ is given by the measure evolution equation
\begin{align*}
d\la f,\bar \mu_t\ra_E &=  \left\{\la \genL_1f,\bar \mu_t\ra_E+ \la \QQ,\bar \mu_t\ra_E
\la \genL_2f,\bar \mu_t\ra_E+\la \genL^{X_{t}}_3 f,\bar
\mu_t\ra_E\right\}dt\nonumber\\
&+\la \genL^{X_{t}}_4 f,\bar \mu_t\ra_E dV_{t},\quad
\forall f\in C^\infty(\PP) \text{ a.s.}
\end{align*}
Suppose there is a solution of the nonlinear SPDE
\begin{align} \label{Eq:NonlinearSPDE}
\begin{aligned}
d\upsilon(t,\pp) &= \left\{\genL_1^*\upsilon(t,\pp) +\genL_3^{*,X_{t}}\upsilon(t,\pp) + \left(\int_{\pp'\in \PP} \QQ(\pp')\upsilon(t,\pp')d\pp'\right) \genL_2^* \upsilon(t,\pp)\right\}dt\\
&+ \genL_4^{*,X_{t}}\upsilon(t,\pp)dV_t,\quad t>0,\quad \pp\in \PP
\end{aligned}
\end{align}
where $\genL_i^*$ denote adjoint operators, with initial condition
\begin{equation*} \lim_{t\searrow 0}\upsilon(t,\pp)d\pp = U(d\pp). \end{equation*}
Then
\begin{equation*}
 \bar{\mu}_t = \upsilon(t,\pp)d\pp.
 \end{equation*}
\end{theorem}

We close this section, by briefly describing the method of moments that leads to the numerical computation of the loss from default. We focus our discussion on the homogeneous case and we refer the reader to \cite{GieseckeSpiliopoulosSowersSirigano2012} for the general case.

Firstly, we remark that the SPDE (\ref{Eq:NonlinearSPDE}) can be supplied with appropriate boundary conditions, which  as it is mentioned in \cite{GieseckeSpiliopoulosSowersSirigano2012}, are
$$\upsilon(t,\lambda = 0) = \upsilon(t, \lambda = \infty) = 0.$$

Secondly, it turns out that for $k\in\N$, the moments $u_k(t) = \int_0^{\infty} \lambda^k \upsilon(t, \lambda) d \lambda$ exist almost surely.  By (\ref{E:firstorder}) is is clear that we want to compute $u_{0}(t)$. In particular, note that the limiting loss $L_t = 1 - u_0(t)$.

  By an integration by parts and using the boundary conditions at $\lambda = 0$ and at $\lambda = \infty$, we can prove that they follow the following system of stochastic differential equations
\begin{equation}\label{Eq: momentSDEone}
\begin{aligned}
d u_k(t)  &= \big\{ u_k(t) \big{(} - \alpha k + \beta^S b_0(X_t) k + 0.5 (\beta^S)^2 \sigma_0^2 (X_t) k (k-1)  \big{)}   \\
&+ u_{k-1}(t) \big{(} 0.5 \sigma^2 k(k-1) + \alpha \bar{\lambda} k + \beta^C k u_1(t) \big{)} - u_{k+1}(t) \big\} dt  + \beta^S \sigma_0(X_t) k u_k(t) d V_t, \\
u_k(0) &= \int_0^{\infty} \lambda^k \Lambda_{\circ} (\lambda) d \lambda,
\end{aligned}
\end{equation}
where $\Lambda_{\circ} (\lambda)=\lim_{N\rightarrow\infty}\frac{1}{N}\sum_{n=1}^{N}\delta_{\lambda_{0}^{\NN}}(\lambda)$.

 The system \eqref{Eq: momentSDEone} is a non-closed system since to determine $u_k(t)$, one needs to know $u_{k+1}(t)$. So, in practice one must perform a truncation at some level $k = K$ where we let $u_{K+1} = u_K$ (that is, we use the first $K+1$ moments). As it is shown in \cite{GieseckeSpiliopoulosSowersSirigano2012} one needs relatively small numbers of moments in order to compute the zero-th moment $u_{0}(t)$ with good accuracy. Then, by solving backwards, one computes $u_0(t)$ and  from this one gets the limiting loss distribution
\[
L_t = 1 - u_0(t).
\]

\section{Central limit theorem correction}\label{S:CLT}

The asymptotics of \eqref{E:SIPDE} give via \eqref{E:firstorder} the limiting behavior of the system as the number of components becomes large.  Starting with that
result, the results in \cite{SpiliopoulosSiriganoGiesecke2013} develop Gaussian fluctuation theory analogous to the central limit theory (see for example \cite{daipra-etal}, \cite{daipra-tolotti}, \cite{FernandezMeleard}, \cite{KurtzXiong} for some related literature). This  result  provides the leading order asymptotics correction to the law of large numbers approximation developed in Section \ref{S:LLN}. In practical terms, the usefulness of such of a result is twofold: (a) the approximation is accurate even for portfolios of moderate size, see  \cite{SpiliopoulosSiriganoGiesecke2013}, and (b): one can make use of the approximation to develop tractable statistical inference procedures for the statistical calibration of such models, see \cite{SirignanoSchwenklerGiesecke2013}.

To be more precise, let us define the signed measure
\begin{equation*}
\Xi^N_t = \sqrt{N}\lb \mu^N_t-\bar{\mu}_t\rb;
\end{equation*}
as $N\to \infty$.  Conditional on the exogenous systematic risk process $X$, a central limit theorem applies and $\bar{\Xi}= \lim_{N\to \infty}\Xi^N$ exists  in an appropriate space of distributions and is Gaussian.   Unconditionally, it may not be Gaussian but is of mean zero (since we have removed the bias $\bar \mu$ from $\mu^N$).

The usefulness of the fluctuation analysis is that it leads to a second-order approximation to the distribution of the portfolio loss $L^{N}$ in large pools. The fluctuations analysis yields an approximation which improves the first-order approximation \eqref{E:firstorder} suggested by the LLN, especially for smaller system sizes $N$.

In particular, Theorem \ref{T:MainCLT} implies that
\begin{equation*}\label{ApproxMain00}
\mathbb{P}(\sqrt{N}(L^{N}_{t}-L_{t})\geq \ell)\approx \mathbb{P}(\bar{\Xi}_{t}(\PP)\leq-\ell)
\end{equation*}
for large $N$. This motivates the approximation
\begin{align*}
\mu^N_t = \frac{1}{\sqrt{N}}\Xi^N_{t} + \bar{\mu}_t \overset{d} \approx \frac{1}{\sqrt{N}} \bar{\Xi}_{t} + \bar{\mu}_t,
\end{align*}
which then implies the following second-order approximation for the portfolio loss.
\begin{align}\label{ApproxMain}
L_t^N \overset{d} \approx L_t - \frac{1}{\sqrt{N}} \bar{\Xi}_{t}(\PP).
\end{align}

The numerical computation of the second-order approximation \eqref{ApproxMain} suggested by the fluctuation analysis is amenable to a moment method similar to that used for computing the first-order approximation (\ref{E:firstorder}).  In addition to solving the LLN SIPDE, we would also need to solve for the fluctuation limit. This limit is governed by a stochastic evolution equation, which gives rise to an additional system of ``fluctuation moments.'' This system is driven by the exogenous systematic risk process $X$ and the martingale $\bar{\mart}_{t}$ in Theorem \ref{T:MainCLT} that is conditionally Gaussian given $X$.

Left of Figure \ref{fig:CLT} compares the approximate loss distribution with the actual loss distribution for specific parameter choices.  It is evident from the numerical comparisons that  the second-order approximation has increased accuracy, especially for smaller portfolios and in the tail of the distribution.
The right of Figure \ref{fig:CLT} compares for the $95$ and $99$ percent value at risk (VaR) between the actual loss, LLN approximation (\ref{E:firstorder}), and approximation (\ref{ApproxMain}) for a pool of $N = 1,000$ names.  It is also evident from the figure that the approximation for the VaR based on (\ref{ApproxMain}) is much more accurate than the law of large numbers approximation.

\begin{figure}[t!]
\begin{center}
\includegraphics[scale=0.5]{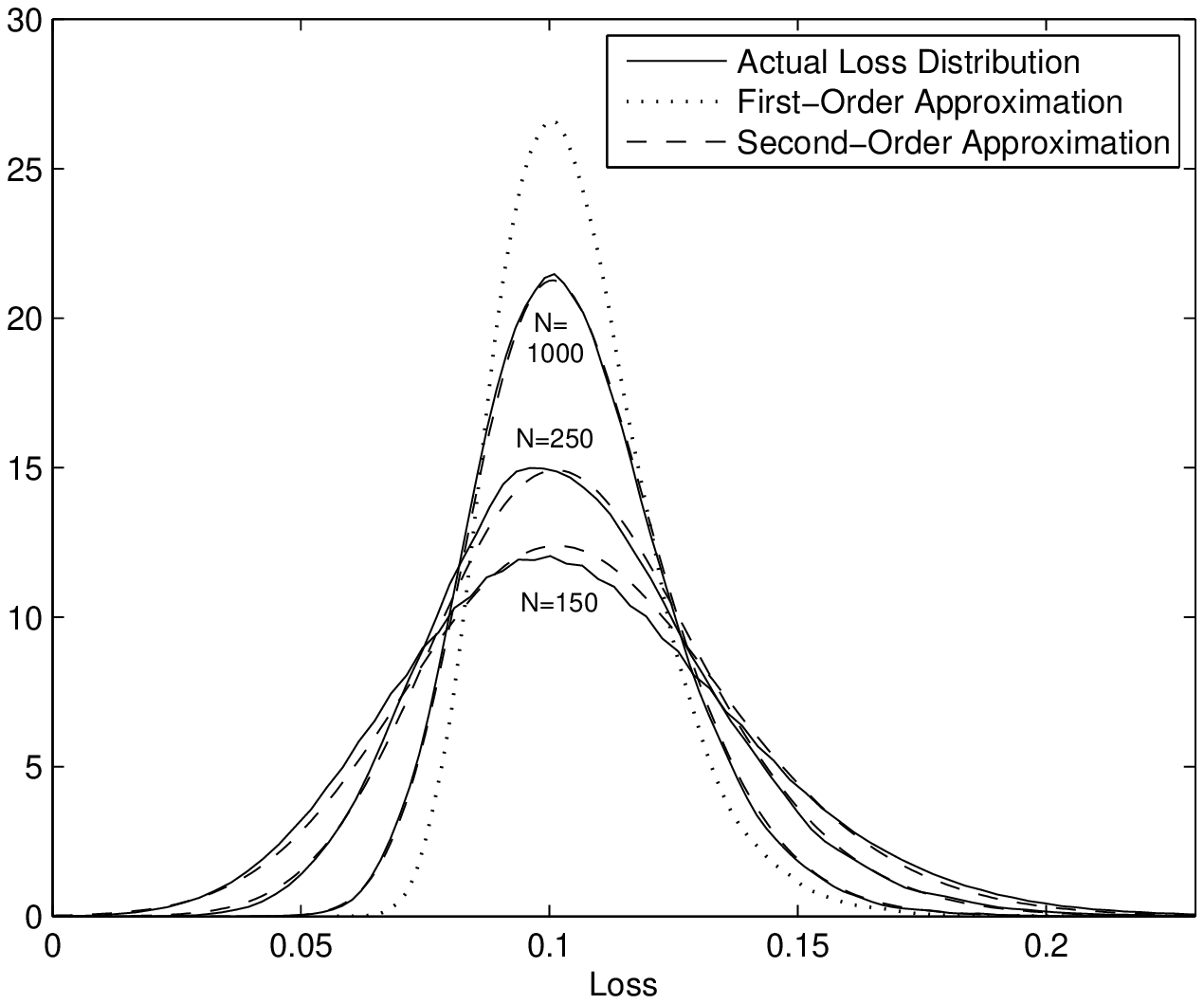}
\includegraphics[scale=0.5]{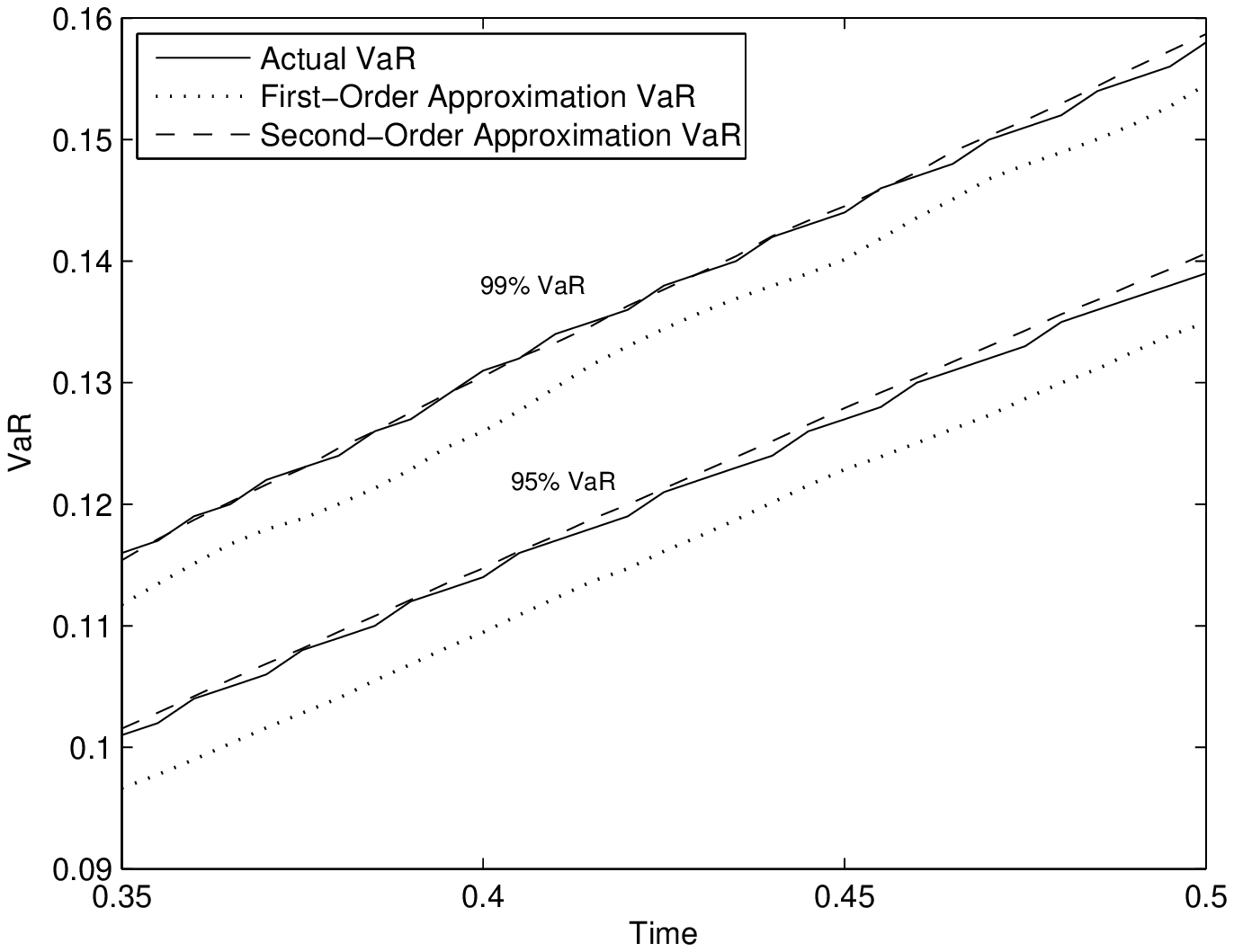}
\end{center}
\caption{\label{fig:CLT} On the left: Comparison of approximate and actual loss distributions of failure rate $\Fail^N_t$ for different $N$ at $t = 0.5$. Parameter choices:
$(\sigma,\alpha,\bar{\haz},\haz_0,\beta^C,\beta^S)=(.9,4,.2,.2,1,1)$. On the right: Comparison of approximate and actual VaR.  Parameter choices:
$(\sigma,\alpha,\bar{\haz},\haz_0,\beta^C,\beta^S)=(.9,4,.2,.2,1,1)$. In both cases, $X$ is an OU process with reversion speed 2, volatility 1, initial value 1 and mean 1.
 }
\end{figure}

Let us close this section, with a few words on the actual mathematical result. It turns out that the convergence $\bar{\Xi}= \lim_{N\to \infty}\Xi^N$ happens in an appropriate weighted Hilbert space, which we denote by  $W^{J}_{0}(w,\rho)$, with $w$ and $\rho$ the appropriate weight functions, $J\in\N$ and  $W^{-J}_{0}(w,\rho)$ will be its dual. Such weighted Sobolev spaces were introduced in \cite{Purtukhia1984} and further generalized in  \cite{GyongiKrylov1992} to study stochastic partial differential equations with unbounded coefficients. These weighted spaces turn out to be convenient for the present situation, see \cite{SpiliopoulosSiriganoGiesecke2013}.

In order to state the convergence result, we introduce some operators. Let $\pp\in \PP\subset \mathbb{R}^{6}$ and for $f\in C^{2}_{b}(\PP)$, define
\begin{align*}\label{E:Operators2}
(\mathcal{G}_{x,\mu}f)(\pp)&=(\genL_1f)(\pp)+ (\genL^{x}_3 f)(\pp)+ \la \QQ, \mu\ra (\genL_2f)(\pp)+\la \genL_2f, \mu\ra \QQ(\pp)\\
(\genL_5 (f,g))(\pp) &=\sigma^{2} \frac{\partial f}{\partial \lambda}(\pp)\frac{\partial g}{\partial \lambda}(\pp)\lambda\\
(\genL_6 (f,g))(\pp) &=f(\pp) g(\pp)\lambda\\
(\genL_7 f)(\pp) &=f(\pp)\lambda\\
\end{align*}

Then, we have the following theorem related to the fluctuations analysis.
\begin{theorem}\label{T:MainCLT}[Theorem 4.1 in \cite{SpiliopoulosSiriganoGiesecke2013}]
For $J>0$ large enough and for appropriate weight functions $(w,\rho)$, the sequence $\{\Xi^{N}_{t}, t\in[0,T]\}_{N\in \N}$ is relatively compact in $D_{W_{0}^{-J}(w,\rho)}[0,T]$. For any $f\in W_{0}^{J}(w,\rho)$, the limit accumulation point of $\Xi^{N}$, denoted by $\bar{\Xi}$, is unique in $W_{0}^{-J}(w,\rho)$ and  satisfies the  stochastic evolution equation
\begin{equation}\label{Eq:CLT}
\la f,\bar \Xi_t\ra = \la f,\bar \Xi_0\ra+ \int_{0}^{t}\la \mathcal{G}_{X_{s},\bar{\mu}_{s}}f,\bar
\Xi_s\ra ds+\int_{0}^{t}\la \genL^{X_{s}}_4 f,\bar \Xi_s\ra dV_{s}+\la f,\bar{\mart}_t\ra, \text{ a.s.}
\end{equation}
for any $f\in W_{0}^{J}(w,\rho)$, where $\bar{\mart}$ is a distribution-valued martingale with predictable variation process
\begin{equation*}
[\la f,\bar{\mart}\ra]_t=\int_{0}^{t}\left[\la \genL_5 (f,f),\bar \mu_s\ra+ \la \genL_6 (f,f),\bar \mu_s\ra+\la \genL_2 f,\bar \mu_s\ra^{2}\la \QQ,\bar \mu_s\ra-2\la \genL_7 f,\bar \mu_s\ra\la \genL_2 f,\bar \mu_s\ra\right]ds.
\end{equation*}
Conditional on the $\sigma$-algebra  $\mathcal{V}_{t}$ that is generated by the $V-$Brownian motion, $\bar{\mart}_t$ is centered Gaussian with covariance function, for $f,g\in W_{0}^{J}(w,\rho)$,
  given by
\begin{align}
\mathrm{Cov}\left[\la f, \bar{\mart}_{t_1}\ra, \la g, \bar{\mart}_{t_2}\ra\,\Big|\, \mathcal{V}_{t_1 \vee t_2}\right]&=\BE\bigg[\int_{0}^{t_1 \wedge t_2}\left[\la \genL_5 (f,g),\bar \mu_s\ra+ \la \genL_6 (f,g),\bar \mu_s\ra+\la \genL_2 f,\bar \mu_s\ra\la \genL_2 g,\bar \mu_s\ra\la \QQ,\bar \mu_s\ra\right.\nonumber\\
& \hspace{1cm}\left.-\la \genL_7 g,\bar \mu_s\ra\la \genL_2 f,\bar \mu_s\ra-\la \genL_7 f,\bar \mu_s\ra\la \genL_2 g,\bar \mu_s\ra\right]ds\, \Big|\,\mathcal{V}_{t_1 \vee t_2}\bigg].\label{Eq:ConditionalCovariation}
\end{align}
\end{theorem}

It is clear that if $\beta^{S}_{\NN}=0$ for all $\NN$, then the limiting distribution-valued martingale $\bar{\mart}$ is centered
Gaussian with covariance operator given by the (now deterministic) term within the expectation in (\ref{Eq:ConditionalCovariation}).

The main idea for the derivation of (\ref{Eq:CLT}) comes from the proof of the convergence to the solution of \eqref{E:weakSIPDE}.  Define
\begin{equation*} (\genL^\circ_1 f)(\pp) = \frac12 \sigma^2\haz\frac{\partial^2 f}{\partial \haz^2}(\pp) - \alpha(\haz-\bar \haz)\frac{\partial f}{\partial \haz}(\pp)\end{equation*}
for $\pp = (\haz,\alpha,\bar \haz,\sigma,\beta^C,\beta^S)$.  Let's also assume for the moment that $\beta^S_\NN=0$ for every $n\in\N$, i.e, let's neglect exposure to the exogenous
risk $X$ and focus on contagion.  Then we can write
the evolution of $\la f,\mu^N_t\ra$ as
\begin{align*} d\la f,\mu^N_t\ra
&= \frac{1}{N}\sum_{n=1}^N \genL^\circ f(\pp^N_t)1_{\{t<\tau_\NN\}}dt - \frac{1}{N}\sum_{n=1}^N f(\pp^\NN_t)\haz^\NN_t 1_{\{t<\tau^\NN\}}dt \\
&\qquad + \frac{1}{N}\sum_{n=1}^N \sum_{m=1}^N \lb f\left(\pp^\NN+\frac{\beta^C_\NN}{N}e_1\right)-f(\pp^\NN_t)\rb 1_{\{\tau_\NN<t\}}\haz^m_t 1_{\{\tau_m\le t\}}dt + dM_t\\
&\approx \la \genL_1 f,\mu^N_t\ra dt + \la \genL_2 f,\mu^N_t\ra \la \QQ,\mu^N_t\ra dt + dM_t
\end{align*}
where $M$ is a martingale which may change from line to line.  This leads to \eqref{E:weakSIPDE}, when $\beta^S_\NN=0$ for every $n\in\N$, see \cite{GieseckeSpiliopoulosSowers2011}.

To get the Gaussian correction, we see that
\begin{equation*} d\la f,\Xi^N_t\ra \approx \la \genL_1 f,\Xi^N_t\ra + \lb \la \genL_2 f,\Xi^N_t\ra \la\QQ,\mu^N_t\ra + \la \genL_2f,\bar \mu_t\ra \la \QQ,\Xi^N_t\ra\rb dt + dM_t \end{equation*}
where $M$ is a martingale.  For large $N$, $M$ should be Gaussian, in which case $\Xi^N$ is indeed a Gaussian process.
Putting the systematic risk process $X$ back into \eqref{E:maina}--\eqref{E:mainb}, one recovers the result of Theorem \ref{T:MainCLT}.

\section{Analysis of tail events: Large deviations}\label{S:LDP}

Once we have identified what is typical, we can study the structure of atypically large failure rates.
Large deviations outlines a circle of ideas and calculations for understanding the origination and transformation of rare events
(see \cite{MR722136}, \cite{MR758258}).
Large deviation arguments allow us to identify the ``dominant'' way that rare events will occur in complex systems. This is the feature that is being exploited in \cite{SpiliopoulosSowers2013}, i.e., how different sources of stochasticity can lead to system collapse.

By the discussion in Section \ref{S:LLN}, we have that
the pool has a default rate $L_{T}=1-\bar{\mu}_T(\PP)$ at time
$T$.  Let's fix $\ell>L_{T}$. Then $\lim_{N\to \infty}\BP\{L^N_T\ge
\ell\}=0$; it is a \emph{rare event} that the default rate in the
pool exceeds $\ell$.  We want to understand as much as possible
about $\{L^N_T\ge \ell\}$.

Using, the theory of large deviations, we can understand both how rare this event is, and what the
``most likely'' way is for this rare event to occur.   Events far
from equilibrium crucially depend on how rare events propagate
through the system. Large deviations gives rigorous ways to
understand these effects, and we want to use this machinery to
understand the structure of atypically large default clusters in the
portfolio.  A reference for large deviations is \cite{MR1619036}.

If we have that
\begin{equation*} \BP\{L^N\approx \varphi\} \approx \exp\left[-NI(\varphi)\right], \text{ as } N\rightarrow\infty \end{equation*}
for some appropriate functional $I$, then by the contraction principle we should have that
\begin{equation}\label{E:varprob0} \BP\{L^N_T \approx \ell\}\approx \exp\left[-N I'(\ell)\right], \text{ as } N\rightarrow\infty \end{equation}
where
\begin{equation}\label{E:varprob} I'(\ell) = \inf\{I(\varphi): \varphi(T)=\ell\} \end{equation}
(in other words, $I'$ is the large deviations rate function for $L^N_T$).
This gives us the rate at which the tail of the default rate $L^N_T$ decays as the diversification parameter grows.  More importantly,
though, the variational problem \eqref{E:varprob} gives us the \emph{preferred} way which
atypically large default rates occur.  Namely, if there is a $\varphi^*_\ell:[0,T]\to [0,1]$ such that
\begin{equation*}\label{E:optpath} I'(\ell) = I(\varphi^*_\ell) \end{equation*}
then for any $\delta>0$, the Gibbs conditioning principle suggests that
\begin{equation*}\label{E:Gibbs} \lim_{N\to \infty}\BP\{\|L^N_T-\varphi^*_\ell\|\ge \delta| L^N_T\ge \ell\} =0. \end{equation*}

Insights into large deviations of \eqref{E:maina}--\eqref{E:mainb} have been developed in \cite{SpiliopoulosSowers2013} when $\eps\downarrow 0$ and when $\eps=O(1)$ as $N\nearrow\infty$.
We note here that in the case $\eps=O(1)$, the large deviations principle is conditional on the systematic risk $X$. Such results  allow us to study the comparative effect of the systematic risk process $X$ and of the contagion feedback on the tails of the loss distribution.

Before presenting the result, let us first investigate numerically a test case, which is indicative of the kind of results that large deviations theory can give us. Apart from approximating the tail of the distribution, large deviations can give quantitative insights into the most likely path to failure of a system.

For presentation purposes and for the rest of this section, we assume that $\eps=\eps_{N}\downarrow 0$ as $N\uparrow\infty$. Consider a heterogeneous  test portfolio composed initially of $N=200$ names. Let us assume that we can separate the names in the portfolio into three types: Type A is $16.67\%$ of the names, Type B is $33.33\%$ of the names and Type C is $50\%$ of the names. For presentation purposes, we assume that all parameters but the contagion parameter are the same among the different types. In particular, we have the following choice of parameters.
\begin{table}[ht!]
\centering
\begin{tabular}{|l|c|c|c|c|c|c|c|}
\hline
 &  $\alpha$ & $\bar\lambda$ & $\sigma$ & $\lambda_{0}$ & $\gamma$ & $\beta^S$ & $\beta^C$  \\ \hline\hline
Type A  & 0.5 & 2 & 0.5 & 0.2 & 1 & 1 & 10 \\ \hline
Type B  & 0.5 & 2 & 0.5 & 0.2 & 1 & 1 & 3 \\ \hline
Type C  & 0.5 & 2 & 0.5 & 0.2 & 1 & 1 & 1 \\ \hline
\end{tabular}
\caption{\label{parameters} Parameter values for a test portfolio composed of three types of assets. We take $\eps_{N}=\frac{1}{\sqrt{N}}$.}
\end{table}

It is instructive to compare the different cases, based on whether there are contagion effects in the default intensities or not. In particular, we compare two different cases,
 (a) Systematic risk only: $\beta^S\not=0, \beta^C=0$, and (b) Systematic risk and contagion: $\beta^S\not=0, \beta^C\not=0$. In each case, the time horizon is $T=1$.

Using the methods of Section \ref{S:LLN}, one can compute that the typical loss in such a pool at time $T=1$. If contagion effects are not present, i.e., if $\beta^C_{A}=\beta^C_{B}=\beta^C_{C}=0$, then the typical loss in such a portfolio at time $T=1$  is $L_{T}=42.5\%$. If on the other hand, contagion (feedback) effects are present and the $\beta^C$ parameters take the values of Table \ref{parameters}, then the typical loss in such a portfolio at time $T=1$  has been increased to  $L_{T}=72.1\%$. In Figure \ref{fig:LDP_RateFunction}, we plot the large deviations rate functions for each of the two different cases. As we saw in the beginning of this section, the rate function governs the asymptotics of the tail of the loss distribution. Notice that in every case, the rate function is convex and it becomes zero at the corresponding law of large numbers.

\begin{figure}[ht!]
\begin{center}
\includegraphics[width=7 cm, height=12 cm, angle=-90]{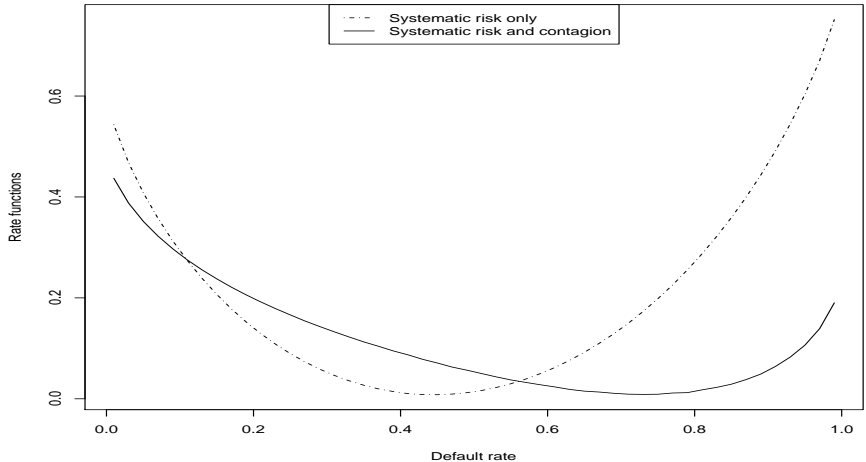}
\caption{\label{fig:LDP_RateFunction} Rate function governing the log-asymtptotics of the tail of the loss distribution. }
\end{center}
\end{figure}

Moreover, since the contagion parameter of Type A is higher than the contagion parameter for Type B or C, one expects that names of Type A will be more prompt to the contagious impact of defaults. Indeed, after computing the rate function and the associated extremals, as defined by large deviations theory, one gets the most likely paths to failure as seen in Figures \ref{fig:LDP1}-\ref{fig:LDP2}. The $\varphi(t)$ trajectories correspond to the contagion extremals for each of the three types, whereas the $\psi(t)$ corresponds to the systematic risk extremal.

\begin{figure}[ht!]
\begin{center}
\includegraphics[width=7 cm, height=12 cm, angle=-90]{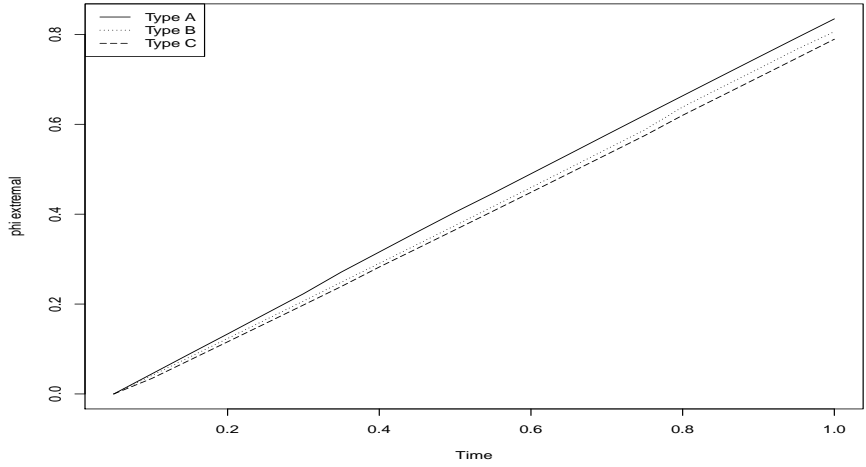}
\caption{\label{fig:LDP1} Optimal $\varphi(t)$ trajectories for the three different types in the pool for $t\in[0,1]$ and $\ell=0.81$. }
\end{center}
\end{figure}

\begin{figure}[ht!]
\begin{center}
\includegraphics[width=7 cm, height=12 cm, angle=-90]{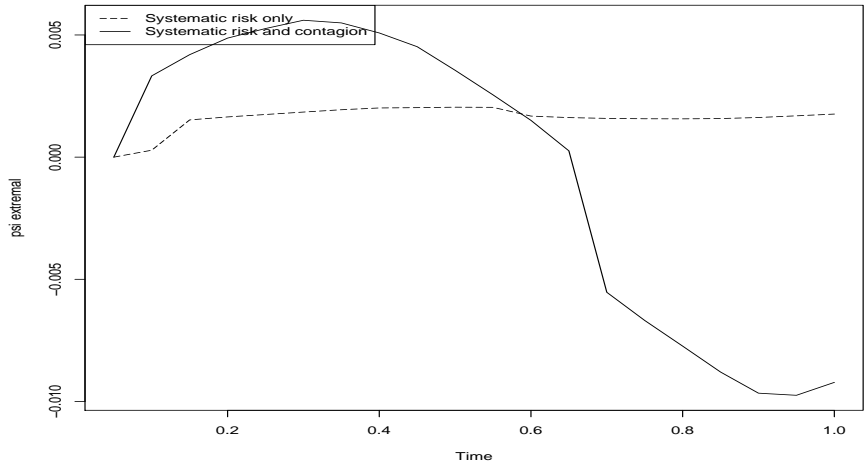}
\caption{\label{fig:LDP2} Comparing optimal $\psi(t)$ trajectories in the case of absence and presence of the contagion effects for $t\in[0,1]$ and $\ell=0.81$. }
\end{center}
\end{figure}

One can make two conclusions out of Figures \ref{fig:LDP1}-\ref{fig:LDP2}. The first conclusion is related to the $\varphi$ extremals (Figure \ref{fig:LDP1}).
We notice that at any given time $t$, the extremal for Type A is bigger than the extremal for
Type B, which in turn is bigger than the extremal of Type C. This implies that unlikely large losses for components of Type A are more likely than unlikely large losses for components of Type B, which are more likely than large losses for components of Type C. Thus, components of
Type A affect the pool more than components of Type B, which in turn affect the pool more than components of Type C even though Type A composes $16.67\%$ of the pool, whereas Type B, composes $33.33\%$ of the pool and Type C composes $50\%$ of the pool. The second conclusion is related to the $\psi$ extremals (Figure \ref{fig:LDP2}). We notice that the effect of the systematic risk is most profound in the beginning but then its significance decreases.

Namely, if a large cluster were to occur, the systematic risk factor is likely to play an important role in the beginning, but then the contagion effects become more important. Assets of Type A are likely to contribute to the default clustering effect more, followed by assets of Type B and the ones that will contribute the least to the default cluster are assets of Type C.

As it is also seen in the numerical experiments done in \cite{SpiliopoulosSowers2013}, the large deviations analysis help quantify the effect that the contagion and  the systematic risk factor have on the behavior of the extremals (the most likely path to failure). An understanding of the role of the preferred paths to  large default rates and the most likely ways in which contagion and systematic risk combine to lead to large default rates would give useful insights into how to optimally hedge against such events.

Let us next proceed by motivating the development of the large deviations principle for the default timing model \eqref{E:maina}--\eqref{E:mainb} that is considered in this paper.

We denote scenarios, i.e., defaults, that are not in $[0,T]$ by an abstract point $\star$ not in $[0,T]$ and define the Polish space
\begin{equation*} \TInt= [0,T]\cup \{\star\} \end{equation*}

To motivate things, let's first assume for simplicity that $\beta^C=\beta^S=0$ and that the system is homogeneous, i.e., that
$\pp^\NN=\pp$ for all $n$. Define
\begin{equation*} \label{E:main2}
d\haz_t = -\alpha (\haz_t-\bar \haz)dt + \sigma \sqrt{\haz_t}dW_t \qquad t>0
\end{equation*}
with $\haz_0=\haz_\circ$.
This Feller diffusion will represent the conditional intensity of a ``randomly-selected'' component of our (homogeneous and independent) system.  Define the
measure $\mu_0\in\Pspace(\R_+)$ by setting
\begin{equation*}\label{E:hazz}
\mu_{0}[0,t] = 1- \BE\left[\exp\left[-\int_{0}^t  \haz_s ds\right]\right]
\end{equation*}
for all $t>0$; $\mu_0$ is the common law of the default times $\tau_n$'s.

In the independent case, i.e., when $\beta^C=0$,  standard Sanov's theorem \cite{MR1619036}, implies that $\{d\Fail^N\}_{N\in \N}$ has a large deviations
principle with rate function
\begin{equation*}
H(\nu,\mu_{0})= \int_{t\in \TInt} \ln \frac{d\nu}{d \mu_0}(t) \nu(dt)
\end{equation*}
if $\nu \ll  \mu_0$ and $H(\nu,\mu_{0})=\infty$ if $\nu\not \ll \mu_0$
(i.e., $H(\nu,\mu_{0})$ is the relative entropy of $\nu$ with respect to $\mu_0$). By the contraction principle, the rate function for $\Fail^{N}_{T}$ is
\begin{equation*}
I^{\text{ind},\prime}(\ell) =\inf\lb H(\nu,\mu_{0}): \text{$\nu\in \Pspace(\R_+)$, $\nu[0,t]=\varphi(t)$ for all $t\in[0,T]$ and $\nu[0,T]=\ell$}\rb\label{Eq:IndependentCaseLDP}
\end{equation*}

In the independent case, we can actually compute both the extremal $\varphi$ that achieves the infimum and the corresponding rate function $I^{\text{ind},\prime}(\ell)$ in closed form.

Assume that $\mu_{0}[0,T]\in (0,1)$ and $\ell\in (0,1)$.  Fix $\nu\in \Pspace(\TInt)$
such that $\nu[0,T]=\ell$.  Define
\begin{equation*} \mu_{0,-}(A)= \frac{\mu_{0}(A\cap [0,T])}{\mu_{0}[0,T]} \qquad \text{and}\qquad \nu_{-}(A) = \frac{\nu(A\cap [0,T])}{\ell} \end{equation*}
for all $A\in \Borel[0,T]$.  Then $\mu_-$ and $\nu_-$ are in $\Pspace[0,T]$.  We can write that
\begin{equation}\label{E:optH} H(\nu,\mu_{0}) = \ell\lb \hbar(\nu_-,\mu_{0,-}) + \ln \frac{\ell}{\mu_{0}[0,T]}\rb + \ln \frac{\nu\{\pt\}}{\mu_{0}\{\pt\}}\nu\{\pt\} \end{equation}
where $\hbar$ is entropy on $\Pspace[0,T]$.  We can minimize the $\hbar$ term by setting $\nu_-=\mu_{0,-}$, and we get that
\begin{eqnarray}\label{E:iprimereduce} I^{\textrm{ind},\prime}(\ell) &=& \ell \ln \frac{\ell}{\mu_{0}[0,T]} + (1-\ell)\ln \frac{1-\ell}{\mu_{0}\{\pt\}}\\
&=&\ell \ln \frac{\ell}{\mu_{0}[0,T]} + (1-\ell)\ln \frac{1-\ell}{1-\mu_{0}[0,T]}.\nonumber
 \end{eqnarray}
This is in fact obvious; $ L^N_T = \frac{1}{N}\sum_{n=1}^N 1_{\{ \tau_n\le T\}}$, and in this case the $1_{\{\tau_n\le T\}}$'s are i.i.d. Bernoulli random variables
with common bias $ \mu_{0}[0,T]$.
The rate function $I^{\textrm{ind},\prime}(\ell)$ of \eqref{E:iprimereduce} is the entropy of Bernoulli coin flips.  Of more interest, however, is the optimal path.
In setting $\nu_-=\mu_-$ in \eqref{E:optH}, we essentially identify the optimal path
\begin{equation*}\label{E:optimalpath}
\varphi(t) = \ell \frac{\mu_{0}[0,t]}{\mu_{0}[0,T]},
\end{equation*}
where the last relation holds since we also require $\varphi(T)=\ell$.

It turns out that one can extend this result to give a generalized Sanov's theorem for the case $\beta^C> 0$,
where $d\Fail^N$ feeds back into the dynamics of the $\haz^n$'s. The case $\beta^{S}>0$ can be treated using a conditioning argument and the well developed theory of large deviations for small noise diffusions. For the heterogeneous case, one needs an additional variational step which minimizes over all the possible  ways that losses are distributed among systems of different types.
Even though an explicit closed form expression for the extremals and for the corresponding rate function is no longer possible, one can still rely on numerically computing them.
Let us make this discussion precise.

To fix the discussion, let us  assume (see \cite{SpiliopoulosSowers2013} for the general case) that the exogenous risk $X$ is of Ornstein-Uhlenbeck type, i.e.,
\begin{equation*}\begin{aligned}
dX_t&=-\gamma X_tdt+dV_t\\
X_0&=x_{\circ}\end{aligned}\label{Eq:SpecificExogeneousRisk}
\end{equation*}

Let $W^*$ be a reference Brownian motion.  Fix  a name in the pool $\pp=(\lambda_{\circ},\alpha,\bar \lambda,\sigma,\beta^C,\beta^S)\in \Types$ and time horizon $T>0$.

The Freidlin-Wentzell theory of large deviations for SDE's gives us a natural starting point.   In
the Freidlin-Wentzell analysis, a dominant ODE is subjected to a small diffusive perturbation; informally, the Freidlin-Wentzell theory
tells us that if we want to find the probability that the randomly-perturbed path is close to a reference trajectory, we should use that reference trajectory in the dynamics.
This leads to the correct LDP rate function for the original SDE.  If we want to find the asymptotics
of the probability that $\left(d\Fail^N\approx d \varphi, \eps_{N} dX\approx d\psi \right)$ for some absolutely continuous functions $\varphi$ and $\psi$, i.e., $\varphi,\psi\in AC\left([0,T],\mathbb{R}\right)$, we should consider the stochastic hazard functions
\begin{equation*} \label{E:DuffiePanSingletonequation} \begin{aligned}
 d\lambda^{\varphi,\psi}_t &=  -\alpha (\lambda^{\varphi,\psi}_t-\bar \lambda)dt + \sigma \sqrt{\lambda^{\varphi,\psi}_t}dW^*_t +\beta^C d\varphi(t) + \beta^S \lambda^{\varphi,\psi}_t d\psi(t) \qquad t\in [0,T]\\
\lambda_0 &= \lambda_\circ. \end{aligned}\end{equation*}
This will represent the conditional intensity of a ``randomly-selected'' name in our pool.
Define next
\begin{equation*}\label{E:fdef} f_{\varphi,\psi}^{{\pp}}(t) = \BE\left[\lambda^{\varphi,\psi}_t\exp\left[-\int_{s=0}^t \lambda^{\varphi,\psi}_s ds\right]\right], \end{equation*}
where, we have used the superscript $\pp$ to denote the dependence on the particular type.  Then for every $t\in [0,T]$ we have that
\begin{equation*}\label{E:fdens} \int_{s=0}^t f_{\varphi,\psi}^{{\pp}}(s)ds = 1-\BE\left[\exp\left[-\int_{s=0}^t \lambda^{\varphi,\psi}_s ds\right]\right] = \BP\lb \int_{s=0}^t \lambda^{\varphi,\psi}_s ds>\ee\rb \end{equation*}
where $\ee$ is an exponential(1) random variable which is independent of $W^*$.  In other words, $f_{\varphi,\psi}^{{\pp}}$ is the density (up to time $T$) of a default time
whose conditional intensity is $\lambda^{\varphi,\psi}$. In fact, due to the affine structure of the model, we have an explicit expression for $f_{\varphi,\psi}^{{\pp}}$ (see Lemma 4.1 in \cite{SpiliopoulosSowers2013}).

For given trajectories $\varphi$ and $\psi$ in $AC([0,T];\R)$,
define $\mu_{\varphi,\psi}^{{\pp}}\in \Pspace(\TInt)$ as
\begin{equation*}\label{E:muDef}  \mu_{\varphi,\psi}^{{\pp}}(A) = \int_{t\in A\cap [0,T]}f_{\varphi,\psi}^{{\pp}}(t)dt + \delta_{\pt}(A)\lb 1- \int_{0}^T f_{\varphi,\psi}^{{\pp}}(t)dt\rb  \end{equation*}
for all $A\in \Borel(\TInt)$.

At a heuristic level one can derive the large deviations principle as follows. Let us assume that we can establish that
\begin{equation*}
\BP\{L^N \approx \varphi|X^N \approx \psi\} \approx \exp\left[-N I^{\circ}(\varphi,\psi)\right]
\end{equation*}
and that  $\left\{ X^N_{\cdot}=\eps_{N}X_{\cdot}, N<\infty\right\}$ also has large deviations principle in $C([0,T];\mathbb{R})$ with action functional $J_{X}$; i.e.,
\begin{equation*}
\BP\lb X^N\approx \psi\rb \approx \exp\left[-\frac{1}{\eps^2_N}J_{X}(\psi)\right]
\end{equation*}
as $N\nearrow \infty$. Then, we should have that
\begin{equation*}\label{E:genLDP}
 \BP\{L^N\approx \varphi,\, X^N \approx \psi\} \approx \exp\left[-N I^{\circ}(\varphi,\psi) - \frac{1}{\eps_N^2}J_{X}(\psi)\right].
 \end{equation*}
In fact, the previous heuristics can be carried out rigorously and in the end one derives the following rigorous large deviations result.
\begin{theorem}[Theorem 3.8 in \cite{SpiliopoulosSowers2013}]\label{T:LDPHeterogeneous2}
Consider the system defined in \eqref{E:maina}-\eqref{E:mainb} with $\lim_{N\rightarrow\infty}\eps_{N}=0$ such that $\lim_{N\rightarrow\infty}N\eps^{2}_{N}= c\in(0,\infty)$ and let $T<\infty$. Under the appropriate assumptions the family $\{L^{N}_{T},N\in\mathbb{N}\}$ satisfies the large deviation principle, with rate function
\begin{align*}
I^{\prime}(\ell)&=\inf\left\{ I(\varphi,\psi): \varphi\in C\left({\Types}\times [0,T]\right),\psi\in C\left([0,T]\right), \psi(0)= \varphi({\pp},0)=0,\right.\nonumber\\
&\qquad\qquad\left.\bar{\varphi}(s)=\int_{{\Types}}\varphi({\pp},s)U(d{\pp}), \bar{\varphi}(T)=\ell \right\}
\end{align*}
where if  $\varphi\in AC\left({\Types}\times [0,T]\right),\psi\in AC\left([0,T]\right), \psi(0)=\varphi({\pp},0)=0$, then
\begin{equation*} \label{E:exLDP}
\begin{aligned}
I(\varphi,\psi)&=\int_{{\Types}}H\left(\varphi({\pp}),\mu^{{\pp}}_{\bar{\varphi},\psi}\right)U(d{\pp})+\frac{1}{c}J_{X}(\psi)
\end{aligned}
\end{equation*}
and $I(\varphi,\psi)=\infty$ otherwise. Here, $J_{X}(\psi)$ is the rate function for the process $\{\eps_{N} {X}^{N}, N<\infty\}$. Namely, for $\psi\in AC\left([0,T];\R\right)$ with $\psi(0)=0$ we have
\begin{equation*}
J_{X}(\psi)=\frac{1}{2}\int_{0}^{T}\left|\dot{\psi}(s)+\gamma \psi(s)\right|^{2}ds
\end{equation*}
and $J_{X}(\psi)=\infty$ otherwise.  $I^{\prime}(\ell)$ has compact level sets.
\end{theorem}

If the heterogeneous portfolio is composed by $K$ different types of assets with homogeneity within each type, then Theorem \ref{T:LDPHeterogeneous2} simplifies to the following expression.

For $\xi,\varphi,\psi\in AC([0,T])$ let us define the functional
\[
g^{{\pp}}(\xi,\varphi,\psi)=\int_{0}^T \ln\left(\frac{\dot \xi(t)}{f^{{\pp}}_{\varphi,\psi}(t)}\right)\dot \xi(t)dt+ \ln \left(\frac{1-\xi(T)}{1-\int_{0}^T f^{{\pp}}_{\varphi,\psi}(t)dt}\right)\left(1- \xi(T)\right)
\]
Due to the affine structure of the model, we have an explicit expression for $f_{\varphi,\psi}^{{\pp}}$ (see Lemma 4.1 in \cite{SpiliopoulosSowers2013}).

Assume that $\kappa_{i}\%$ of the names are of type $A_{i}$ with $i=1,\cdots, K$ and $\sum_{i=1}^{K}\kappa_{i}=100$.
Setting $\varphi({\pp},s)=\sum_{i=1}^{K}\frac{\kappa_{i}}{100}\varphi_{A_{i}}(s)\chi_{\{{\pp}_{A_{i}}\}}(\pp)$, we get the following simplified expression for the rate function
\begin{align*}
I^{\prime}(\ell)&=\inf\left\{\sum_{i=1}^{K}\frac{\kappa_{i}}{100}g^{{\pp}_{A_{i}}}(\varphi_{A_{i}}, \varphi,\psi)+\frac{1}{c}J_{X}(\psi): \varphi(t)=\sum_{i=1}^{K}\frac{\kappa_{i}}{100}\varphi_{A_{i}}(t) \textrm{ for every }t\in[0,T]\right.\nonumber\\
&\qquad\qquad\qquad\left. \varphi(T)=\ell, \varphi_{A_{i}}(0)=\psi(0)=0, \varphi_{A_{i}},\psi\in AC([0,T]) \textrm{ for every }i=1,\cdots, K\right\}.\nonumber
\end{align*}

An optimization algorithm can then be employed to solve the minimization problem associated with $I^{\prime}(\ell)$ and compute the extremals $\varphi_{A_{i}}$
for $i=1,\cdots,K$ and $\psi$. This is the formula that the numerical example presented in Figures \ref{fig:LDP1}-\ref{fig:LDP2} was based on. In the numerical example that was
considered there we had three types, i.e., $K=3$.

The large deviations results have a number of important applications. Firstly, they lead to an analytical approximation of the tail of the distribution of the failure rate $L^{N}$ for large systems. These approximations complement the first- and second- order approximations  suggested by the law of large numbers and fluctuations analysis of Sections \ref{S:LLN} and \ref{S:CLT} respectively and facilitates the estimation of the likelihood of systemic collapse. Secondly, the large deviations results provide an understanding of the ``preferred" ways of collapse, which can also be used to design ``stress tests" for the system. In particular, this understanding can guide the selection of meaningful stress scenarios to be analyzed. Thirdly, they can motivate the design of asymptotically efficient importance sampling schemes for the tail of the portfolio loss. We discuss some of the related issues in Section \ref{S:IS}.

\section{Monte Carlo methods for estimation of tail events: Importance sampling}\label{S:IS}

Suppose we want to computationally simulate $\BP\{\Fail^N_T\ge \ell\}$, where $\lim_{N\to \infty}\BP\{L^N_T\ge
\ell\}=0$ again holds.
Accurate estimates of such rare-event probabilities are important in many applications areas of our system \eqref{E:maina}--\eqref{E:mainb}, including credit risk management, insurance, communications and reliability. Monte Carlo methods are widely used to obtain such estimates in large complex systems such as ours; see, for example, \cite{bassamboo-jain, bassamboo-juneja-zeevi, carmona-crepey,  vestal-carmona-fouque, sisr,  gkmt, Glasserman,GlassermanLi2005, zhang-etal}.

Standard Monte Carlo sampling techniques perform very poorly in estimating rare events (for which, by definition, most samples can be discarded).  Importance sampling, which involves a change of measure, can be used to address this issue. In general, large deviations theory provides an optimal way to `tilt' measures. The variational problems identified by large deviations usually lead to measure transformations under which pre-specified rare events become much more likely, but which give unbiased estimates of probabilities of interest; see for example \cite{AsmussenGlynn2007,Bucklew2004,DupuisWang2004,DupuisWang2007,DupuisSpiliopoulosWang,GlassermanWang1997,Sadowsky1996}.

Let $\Gamma^N$ be any unbiased estimator of $\BP\{\Fail^N_T\ge \ell\}$
that is defined on some probability space with probability measure
$\BQ$. In other words, $\Gamma^N$ is a random variable such that $\BE^{\BQ}\Gamma^N=\BP\{\Fail^N_T\ge \ell\}$,
where $\BE^{\BQ}$ is the expectation operator associated with $\BE$. In our setting, it takes the form
\begin{equation*}\Gamma^N=1_{\{\Fail^N_T> \ell\}}\frac{d\BP}{d\BQ},\end{equation*}
where $\frac{d\BP}{d\BQ}$ is the associated Radon-Nikodym derivative.

Importance sampling involves the generation of independent copies of
$\Gamma^N$ under $\BQ$; the estimate is the sample mean. The specific
number of samples required depends on the desired accuracy, which is measured
by the variance of the sample mean. However, since the samples are independent
it suffices to consider the variance of a single sample. Because of
unbiasedness, minimizing the variance is equivalent to minimizing the second
moment.
An application of Jensen's inequality, shows that if
\begin{equation*}
\liminf_{N\to \infty}\frac{1}{N}\ln\BE^{\BQ}(\Gamma^N)^{2} = -2I'(\ell),
\end{equation*}
then $\Gamma^N$ achieves this best decay rate, and is said to be
\textit{asymptotically optimal}. One wants to choose $\BQ$ such that asymptotic optimality is attained.

To motivates things let us assume for the moment that $\beta^{C}=\beta^{S}=0$ and that the system is homogeneous, i.e., that
$\pp^\NN=\pp$ for all $n$. In the independent and homogeneous case, $\Xi_{n}=1_{\{\tau_n\le T\}}$ are i.i.d. random variables such that
for every $t\in [0,T]$
\begin{equation*}
\BP\lb \tau_n\le t \rb= \BP\lb \int_{0}^t \lambda^{0,0}_s ds>\ee\rb= 1-\BE\left[\exp\left[-\int_{0}^t \lambda^{0,0}_s ds\right]\right]=\int_{0}^t f_{0,0}(s)ds
 \end{equation*}

For notational convenience, we shall define
\begin{equation*}
p=\int_{0}^T f_{0,0}(s)ds
\end{equation*}

It is easy to see that,
 \begin{equation*}
 NL_{T}^{N}\sim \text{Binomial}(N, p)
 \end{equation*}

 To minimize the variance, we need to increase the probability of defaults. Define

\begin{equation*}
\Lambda^{N}(\theta;t)=\ln \BE\left[e^{\theta L_{t}^{N}}\right]\label{Eq:LogMmt2prelimit}
\end{equation*}

A simple computation shows that
\begin{equation*}
\bar{\Lambda}(\theta;t)=\lim_{N\rightarrow\infty}\frac{1}{N}\Lambda^{N}(N\theta;t)=\ln\left(p\left(e^{\theta}-1\right)+1\right)
\end{equation*}

Define
\begin{equation*}
p_{\theta}=\frac{p e^{\theta}}{1+p (e^{\theta}-1)}
\end{equation*}

Clearly $p_{0}=p$. Notice that the density of a $\text{Binomial}(N, p)$ with respect to a $\text{Binomial}(N, p_{\theta})$ is
\begin{eqnarray*}
\mathcal{Z}_{\theta}&=&\prod_{n=1}^{N}\left(\frac{p}{p_{\theta}}\right)^{\Xi_{n}}\left(\frac{1-p}{1-p_{\theta}}\right)^{1-\Xi_{n}}=\prod_{n=1}^{N}\left[\left(1+p(e^{\theta}-1)\right)e^{-\theta \Xi_{n}}\right]\nonumber\\
&=&e^{N\left(-\theta L^{N}_{T}+\bar{\Lambda}(\theta;T)\right)}
\end{eqnarray*}

Therefore, for $\theta$ fixed,  the suggestion is to simulate under a new change of measure, under which  $NL_{T}^{N}\sim \text{Binomial}(N, p_{\theta})$ and to return the estimator
\begin{equation*}
\Gamma=\frac{1}{M}\sum_{i=1}^{M}1_{\{L^{N,i}_{T}>\ell\}} e^{N\left(-\theta L^{N,i}_{T}+\bar{\Lambda}(\theta;T)\right)}
\end{equation*}

It is clear that this estimator is unbiased. We want to choose $\theta$ that minimizes the variance, or equivalently the second moment. For this purpose, we define the second moment
\begin{equation*}
Q(\ell,\theta)=\BE_{\theta}\Gamma^{2}=\BE_{\theta}\left[1_{\{L_{T}>\ell\}}e^{2N\left(-\theta L^{N}_{T}+\bar{\Lambda}(\theta;T)\right)}\right]
\end{equation*}

Notice that
\begin{equation*}
-\frac{1}{N}\ln Q(\ell,\theta)\geq -2\frac{1}{N}N \left(-\theta \ell+\bar{\Lambda}(\theta;T)\right)=2(\theta \ell-\bar{\Lambda}(\theta;T))
\end{equation*}

Due to convexity of $\bar{\Lambda}(\theta;T)$, we have that the maximizer over $\theta\in[0,\infty)$ of the lower bound is at $\theta^{*}$ such that $\frac{\partial \bar{\Lambda}(\theta^{*};T)}{\partial \theta}=\ell$. In particular, (recall that $\frac{\partial\bar{\Lambda}(0;T)}{\partial\theta}=p$) we have
\begin{equation*}
\theta^{*}=\begin{cases}
\ln\frac{\ell(1-p)}{p(1-\ell)}, & \textrm{ if } \ell>p\\
0, & \textrm{ if } \ell<p
\end{cases}
\end{equation*}

This construction means that under the new measure, we  have
\begin{equation*}
\BP_{\theta^{*}}\lb \tau_n\le T \rb=p_{\theta^{*}}=\ell.
 \end{equation*}
 In fact, we have the following theorem.
\begin{theorem}
Let $\theta^{*}>0$ such that $\frac{\partial \bar{\Lambda}(\theta^{*};T)}{\partial \theta}=\ell$. Then asymptotic optimality holds, in the sense that
\begin{equation*}
\lim_{N\rightarrow\infty}-\frac{1}{N}\ln Q(\ell,\theta^{*})=2 I^{\textrm{ind},\prime}(\ell)
\end{equation*}
where $I^{\textrm{ind},\prime}(\ell)$ is defined in \eqref{E:iprimereduce}.
\end{theorem}
\begin{proof}
By Jensen's inequality we clearly have the upper bound. Namely, for every $\theta\in[0,\infty)$

\begin{equation}
\limsup_{N\rightarrow\infty}-\frac{1}{N}\ln Q(\ell,\theta)\leq 2 I^{\textrm{ind}}(\ell)
\end{equation}

Now, we need to prove that the lower bound is achieved for $\theta=\theta^{*}$, i.e., that

\begin{equation}
\liminf_{N\rightarrow\infty}-\frac{1}{N}\ln Q(\ell,\theta^{*})\geq 2 I^{\textrm{ind}}(\ell)
\end{equation}

Recalling that $\theta^{*}=\ln\frac{\ell(1-p)}{p(1-\ell)}$ and $p=\int_{s=0}^T f_{0,0}(s)ds$, we easily see that
\begin{eqnarray}
\liminf_{N\rightarrow\infty}-\frac{1}{N}\ln Q(\ell,\theta^{*})&\geq& 2\left(\theta^{*}\ell-\bar{\Lambda}(\theta^{*};T)\right)\nonumber\\
&=&2\left(\theta^{*}\ell-\ln\left(p(e^{\theta^{*}}-1)+1\right)\right)\nonumber\\
&=&2\left(\ell\ln\frac{\ell}{p}+(1-\ell)\ln\frac{1-\ell}{1-p}\right)\nonumber\\
&=&2 I^{\textrm{ind}}(\ell)\nonumber
\end{eqnarray}
This concludes the proof of the theorem.
\end{proof}

In the heterogeneous case, i.e. if $\pp^{n}$ can be different for each $n\in\N$, then $NL^{N}_{T}=\sum_{n=1}^{N}1_{\left\{\tau_{n}\leq T\right\}}$ is no longer Binomial, but it is a sum of independent (but not identically distributed) Bernoulli random variables with success probability
\begin{equation*}
p_{n}=\int_{0}^T f^{\pp^{n}}_{0,0}(s)ds
\end{equation*}
indexed by $n$. Due to independence, similar methods as the one described above can be used to construct asymptotically efficient importance sampling schemes in the heterogeneous case.

The scheme just presented essentially amounts to a twist in the intensity of the defaults. However, in contrast to the independent case, i.e., when $\beta^{C}=\beta^{S}= 0$, the situation in the general dependent case $\beta^{C},\beta^{S}\neq 0$ is more complicated. Notice also if at least one one of the $\beta^{C}_{n}$'s is not zero, then the model \eqref{E:maina}--\eqref{E:mainb} does not fall into the category of the doubly-stochastic models, so techniques as the ones used in \cite{bassamboo-jain} do not apply. Also, implementation of interacting particle schemes for Markov Chain models as the ones developed in  \cite{carmona-crepey,  vestal-carmona-fouque} do not readily apply for such intensity models. The re-sampling schemes of \cite{gkmt}  could apply in this setting, but one would need to construct an appropriate mimicking Markov Chain, something which is not clear how to do in the current setting.

We briefly present here an importance sampling scheme for the case that there exists at least one $\beta^{C}_{n}\neq 0$ and also applies independently of whether the systematic effects are present in the model or not.  The suggested measure change essentially mimics the principal idea behind the measure change for the independent case.  To be more precise, one directly twists the intensity of $NL^{N}_{T}=\sum_{n=1}^{N}1_{\left\{\tau_{n}\leq T\right\}}$.

Let $\{S_{k}\}$ be the arrival times of $NL^{N}_{T}$ and notice that $\left\{L^{N}_{T}\geq \ell\right\}=\left\{ S_{\lceil\ell N\rceil}\leq T\right\}$. Let $M^{n}_{s}=1_{\{\tau^{n}>s\}}$ and $\theta^{N}_{s}\geq 1$ be some progressively measurable twisting process. Then, define the measure $\BQ$ via the Radon-Nicodym derivative
\[
Z_{N}=e^{-\int_{0}^{S_{\lceil\ell N\rceil}}\log\left(\theta^{N}_{s-}\right)d(NL^{N}_{s})-\int_{0}^{S_{\lceil\ell N\rceil}}\left(1-\theta^{N}_{s}\right)\sum_{n=1}^{N}\lambda^{n}_{s}M^{n}_{s} ds}.
\]

It is known that if $\BE\left[ e^{-\sum_{k=1}^{\lceil\ell N\rceil}\log\left(\theta^{N}_{S_{k}-}\right)}\right]<\infty$, then $\BQ$ defined by $\frac{d\BP}{d\BQ}=Z_{N}$ is a probability measure and it can be shown that $NL^{N}_{s}$  admits $\BQ-$intensity $\theta^{N}_{s}\sum_{n=1}^{N}\lambda^{n}_{s}M^{n}_{s}$ on the interval $[0, S_{\lceil\ell N\rceil})$.

This construction gives us some freedom into choosing appropriately the twisting process $\theta^{N}_{s}$. Different choices of the twisting process $\theta^{N}_{s}$ are of course possible. For tractability purposes we restrict attention to a one-parameter family and set $$\theta^{N}_{s}=\frac{\beta N}{\sum_{n=1}^{N}\lambda^{n}_{s}M^{n}_{s}}+1.$$

For any $\beta\geq 0$ and under the measure induced by $Z_{N}$, i.e. under $\BQ_{\beta}$, the process $NL^{N}_{s}$ has intensity $\sum_{n=1}^{N}\lambda^{n}_{s}M^{n}_{s}+\beta N$ on  $[0, S_{\lceil\ell N\rceil})$, i.e. it amounts to an additive shift of the intensity. Thus, $\beta$ is a superimposed default rate and its role is to increase the default rate in the whole portfolio.

The purpose then is to optimize the limit as $N\rightarrow\infty$ of the upper bound of the second moment of the resulting estimator over $\beta$.  This is the measure change that is  investigated in \cite{GieseckeShkolnik2011}, and it is shown there that there is a choice of $\beta=\beta^{*}$ for which asymptotic optimality can be established. Namely, there is a choice of $\beta=\beta^{*}$ that minimizes the second moment of the estimator in the limit as $N\rightarrow\infty$. We refer the interested reader to \cite{GieseckeShkolnik2011} for implementation details on this change of measure for related intensity models and for corresponding simulation results.

\section{Conclusions}\label{S:Conclusions}

We presented an empirically motivated model of correlated default timing for large portfolios. Large portfolio analysis allows to approximate the distribution of the
 loss from default, whereas Gaussian corrections make the approximation valid even for portfolios of moderate size. The results can be used to compute
the loss distribution and to approximate portfolio risk measures such as Value-at-Risk or Expected Shortfall. Then, large deviations analysis can help understand the tail of
the loss distribution and find the most-likely paths to systemic failure and to the creation of default clusters. Such results give useful insights into the behavior of systemic risk as a function of the characteristics of the names in the portfolio and can be also potentially used to determine how to optimally safeguard against rare large losses.  Importance sampling techniques can be used to construct asymptotically efficient estimators for tail event probabilities.

\section{Acknowledgements}
The author was partially supported by the National Science Foundation
(DMS 1312124).

\bibliographystyle{alpha}
\newcommand{\etalchar}[1]{$^{#1}$}


\end{document}